\documentclass[letterpaper,10 pt,conference,dvipsnames]{ieeeconf}
\IEEEoverridecommandlockouts

\let\labelindent\relax

\usepackage[absolute,overlay]{textpos}
\usepackage{times}
\usepackage{enumitem}

\usepackage[usenames,dvipsnames]{xcolor}

\usepackage{amsmath,amssymb,mathtools,nccmath}

\usepackage{amsthm}
\usepackage{bm}
\usepackage{bbold,dsfont,bbm}

\usepackage{graphicx}
\usepackage{tikz,tikz-cd,pgfplots}
\pgfplotsset{compat=1.17}

\usepackage[hidelinks]{hyperref}
\usepackage[capitalise,noabbrev]{cleveref}
\usepackage{cite}

\usepackage{multirow}
\usepackage{booktabs}
\usepackage[acronym]{glossaries}

\usepackage{comment}

\usepackage{units}
\usepackage[per-mode=symbol]{siunitx}
\DeclareSIUnit{\eur}{\euro}
\DeclareSIUnit{\usd}{USD}
\DeclareSIUnit{\mph}{mph}
\DeclareSIUnit{\month}{month}
\DeclareSIUnit{\year}{year}
\DeclareSIUnit{\million}{Mil}
\DeclareSIUnit{\mile}{mile}
\DeclareSIUnit{\car}{car}
\DeclareSIUnit{\train}{train}
\DeclareSIUnit{\mmveh}{\text{$\mu$}MV}
\DeclareSIUnit{\nounit}{-}
\usepackage[font=footnotesize]{caption}
\usepackage{subcaption}

\newcommand{\argmin}{\mathop{\mathrm{argmin}}}

\definecolor{lightblue}{rgb}{0.60784,0.76078,0.90196}
\definecolor{darkblue}{rgb}{0.26667,0.44706,0.76863}
\definecolor{lightgreen}{rgb}{0.66275,0.81569,0.55686}
\definecolor{darkgreen}{rgb}{0.43922,0.67843,0.27843}
\definecolor{orange}{rgb}{0.92941,0.49020,0.19216}
\definecolor{yellow}{rgb}{1.00000,0.75294,0.00000}
\definecolor{grey}{rgb}{0.64706,0.64706,0.64706}
\definecolor{purple}{rgb}{0.51373,0.23529,0.04706}
\newacronym{abk:amod}{AMoD}{Autonomous Mobility-on-Demand}
\newacronym{abk:iamod}{\mbox{I-AMoD}}{intermodal \gls{abk:amod}}
\newacronym{abk:av}{\mbox{AV}}{autonomous vehicle}

\newacronym{abk:bpr}{BPR}{Bureau of Public Roads}
\newacronym{abk:bev}{BEV}{Battery Electric Vehicle}
\newacronym{abk:ca}{CA}{congestion-aware}
\newacronym{abk:cara}{CARS}{congestion-aware routing scheme}
\newacronym{abk:cpo}{CPO}{complete partial order}
\newacronym{abk:cdp}{CDP}{co-design problem}
\newacronym{abk:cdpi}{CDPI}{co-design problem with implementation}
\newacronym{abk:dp}{DP}{design problem}
\newacronym{abk:dpi}{DPI}{design problem with implementation}
\newacronym{abk:mdpi}{MDPI}{monotone design problem with implementation}
\newacronym{abk:dcpo}{DCPO}{directed complete partial order}
\newacronym{abk:es}{ES}{e-scooter}
\newacronym{abk:ffcs}{FFCS}{free floating car sharing systems}
\newacronym{abk:fcm}{FCM}{fuel-cell moped}
\newacronym{abk:ghg}{GHG}{greenhouse gas}
\newacronym{abk:icev}{ICEV}{ Internal Combustion Engine Vehicle}
\newacronym{abk:kpi}{KPIs}{Key Performance Indicators}
\newacronym{abk:lw}{LW}{Lightweight}
\newacronym{abk:mm}{{$\mu$}M}{micromobility}
\newacronym{abk:mmveh}{{$\mu$}MV}{micromobility vehicle}
\newacronym{abk:mod}{MoD}{Mobility-on-Demand}
\newacronym{abk:mcdp}{MCDP}{Monotone Co-Design Problem}
\newacronym{abk:mcfp}{MCFP}{multi-commodity flow problem}
\newacronym{abk:moped}{M}{moped}
\newacronym{abk:nyc}{NYC}{New York City}
\newacronym{abk:poset}{poset}{partially ordered set}
\newacronym{abk:ptop}{P2P}{point-to-point}
\newacronym{abk:sb}{SB}{shared bike}
\newacronym{abk:spp}{SPP}{shortest path problem}
\newacronym{abk:kdspp}{k-dSPP}{k-disjoint \gls{abk:spp}}
\newacronym{abk:su}{SU}{Sport Utility}

\newcommand{\power}[1]{\mathcal{P}(#1)}
\newcommand{\prov}{\mathsf{prov}}

\newcommand{\req}{\mathsf{reqs}}

\newcommand{\setOfFunctionalities}[1]{\F{\mathcal{F}_{#1}}}
\newcommand{\setOfFunctionalitiesOp}[1]{\F{\mathcal{F}_{#1}}^{\mathrm{op}}}
\definecolor{dpgreen}{rgb}{0.0, 0.5, 0.0}
\newcommand{\F}[1]{\textcolor{dpgreen}{#1}}
\newcommand{\setOfImplementations}[1]{\mathcal{I}_{#1}}

\newcommand{\tup}[1]{\langle #1 \rangle}
\definecolor{dpred}{rgb}{0.7, 0.0, 0.0}
\newcommand{\R}[1]{\textcolor{dpred}{#1}}
\newcommand{\setOfResources}[1]{\R{\mathcal{R}_{#1}}}

\ifPDFTeX
\newcommand{\tickar}{\begin{tikzcd}[baseline=-0.5ex,cramped,sep=small,ampersand replacement=\&]{}\ar[r,tick]\&{}\end{tikzcd}}
\else
\newcommand{\tickar}{\nrightarrow}
\fi
\newcommand{\op}{^{\mathrm{op}}}

\newcommand{\true}[1]{\mathtt{T}}

\newcommand{\mat}[1]{\mathbf{#1}}
\newcommand{\positionr}{p_r}
\newcommand{\posxr}{x_\mathrm{r}}
\newcommand{\posyr}{y_\mathrm{r}}
\newcommand{\positionf}{p_\mathrm{f}}
\newcommand{\position}{p}
\newcommand{\posxf}{x_\mathrm{f}}
\newcommand{\posyf}{y_\mathrm{f}}
\newcommand{\posx}{x}
\newcommand{\posy}{y}
\newcommand{\posxt}{x_\mathrm{t}}
\newcommand{\posyt}{y_\mathrm{t}}
\newcommand{\heading}{\theta}
\newcommand{\steeringangle}{\delta}
\newcommand{\rearspeed}{v_\mathrm{r}}
\newcommand{\frontspeed}{v_\mathrm{f}}
\newcommand{\steerspeed}{v_\mathrm{s}}
\newcommand{\rearacc}{a_\mathrm{r}}
\newcommand{\angular}{\omega}

\newcommand{\evalue}[1]{\mathbb{E}\left[ #1\right]}
\newcommand{\stateunc}[3]{\mat{P}_{#1|#2}^{#3}}

\definecolor{baiocchi}{RGB}{193,221,245}
\newcommand*\circled[1]{\tikz[baseline=(char.base)]{
            \node[shape=circle,black,inner sep=1pt, fill=baiocchi] (char) {#1};}}

\usetikzlibrary{positioning,intersections, hobby, patterns, calc, decorations.pathmorphing, decorations.markings, shadows,shapes, cd,arrows.meta, fit,quotes}

\tikzset{
   tick/.style={postaction={
      decorate,
      decoration={markings, mark=at position 0.5 with {\draw[-] (0,.4ex) -- (0,-.4ex);}}}
   }
}
\tikzstyle{block} = [draw, rectangle, minimum height=2em, minimum width=3em,font=\bfseries,rounded corners,thick]
\tikzstyle{block} = [draw, rectangle, minimum height=2em, minimum width=3em]
\tikzstyle{block1} = [draw, rectangle, minimum height=1.5em, minimum width=2.5em]
\tikzstyle{blockDyn} = [draw, rectangle, minimum height=2.5em, minimum width=3.5em, align=center, inner sep=10pt, thick, fill=white, copy shadow={draw=black,fill=black,opacity=1,shadow xshift=0.5ex,shadow yshift=-0.5ex}]
\tikzstyle{blockAlg} = [draw, rectangle, minimum height=1.5em, minimum width=2.5em, align=center, inner sep=10pt, thick]
\tikzstyle{sum} = [draw,circle]

\tikzstyle{nodePre} = [circle, draw,inner sep=1pt,node contents={$\preceq$},thick]
\tikzstyle{nodePreEmpty} = [circle, draw,inner sep=1pt,thick]
\tikzstyle{nodePos} = [circle, draw,inner sep=1pt,node contents={$\posceq$},thick]
\tikzstyle{nodeProd} = [rectangle, draw,inner sep=4pt,node contents={$\times$},rounded corners,thick]
\tikzstyle{nodeSum} = [rectangle, draw,inner sep=4pt,node contents={$\mathbf{+}$},rounded corners,thick]

\definecolor{red}{rgb}{0.75, 0.0, 0.0}

\tikzset{fcname/.store in =\fcname, fcname={}}
\tikzset{funame/.store in =\funame, funame={}}
\tikzset{rcname/.store in =\rcname, rcname={}}
\tikzset{runame/.store in =\runame, runame={}}
\tikzset{whereres/.store in =\whereres, whereres=0.5}
\tikzset{wherefun/.store in =\wherefun, wherefun=0.5}
\tikzset{relres/.store in =\relres, relres={above}}
\tikzset{relfun/.store in =\relfun, relfun={above}}
\tikzset{posres/.store in =\posres, posres=1}
\tikzset{posfun/.store in =\posfun, posfun=1}
\tikzset{loos/.store in =\loos, loos=2}
\tikzset{feedback/.store in =\feedback, feedback=0}
\tikzset{
   DP/.style={%
      label/.style={
         font=\everymath\expandafter{\the\everymath\scriptstyle},
         inner sep=5pt,
         node distance=2pt and -2pt},
      semithick,
      node distance=1 and 1,
      rconn/.style={color=white,opacity=0.0,postaction={decorate}, shorten <=3.2pt, shorten >= 0.8,
      decoration={markings, 
      mark= at position 0 with {
               \coordinate (a);
      },
      mark=at position .5 with
      {
              \ifthenelse{\equal{\feedback}{1}}{\def\angleOut{90}\def\angleIn{90}}{\def\angleOut{0}\def\angleIn{180}}    
              \coordinate (b);
              \draw[dashed,dpred,opacity=1.0] (a) to[out=\angleOut,in=\angleIn,looseness=\loos] 
              node[pos=\posres,\relres=\whereres mm,dpred,opacity=1,fill=white,inner sep=1pt,outer sep=1pt]{\footnotesize{\rcname}} (b);
      },
      mark= at position 1 with 
      {
             \ifthenelse{\equal{\feedback}{1}}{\def\angleOut{0}\def\angleIn{0}}{\def\angleOut{180}\def\angleIn{0}} 
              \ifthenelse{\equal{\feedback}{1}}{\def\symbol{\succeq}}{\def\symbol{\preceq}} 
              \coordinate (c);
              \draw[dpgreen,opacity=1.0] (c) to[out=\angleOut,in=\angleIn,looseness=\loos]
              node[pos=\posfun,\relfun=\wherefun mm,dpgreen,opacity=1,fill=white,inner sep=1pt,outer sep=1pt]{\footnotesize{\fcname}} (b){}; %
              \node[draw,circle,inner sep=0.5pt,color=black,fill=white,opacity=1.0] at (b) (nodepreceq) {$\symbol$}; 
      }
      }},
      runconn/.style={color=dpred,dashed,postaction={decorate},
      decoration={markings,
      mark= at position 1 with {
              \coordinate (a);
              \draw[dpred,opacity=1.0,dashed] ($(a)+(0.05,0)$) --++ (0.5,0) node[\relres,pos=\posres]{\footnotesize{\runame}};}
      }
      },
      funconn/.style={color=white,postaction={decorate},
      decoration={markings,
      mark= at position 0 with {
      \coordinate (a);
      \draw[dpgreen] ($(a)+(-0.05,0)$) -- ($(a)+(-0.5,0)$) node[\relfun, pos=\posfun]{\footnotesize{\funame}};}
      }
      },
      execute at begin picture={\tikzset{
         x=\dpx, y=\dpy,
         every fit/.style={inner xsep=\dpx, inner ysep=\dpy}}}
      },
   dpx/.store in=\dpx,
   dpx = 1.5cm,
   dpy/.store in=\dpy,
   dpy = 1.5ex,
   dp port sep/.store in=\dpportsep,
   dp port sep=2,
   dp port length/.store in=\dpportlen,
   dp port length=4pt,
   dp min width/.store in=\dpminwidth,
   dp min width=0.5cm,
   dp rounded corners/.store in=\dpcorners,
   dp rounded corners=2pt,
   dp small/.style={dp port sep=1, dp port length=2.5pt, dpx=.4cm, dp min width=.4cm, dpy=.7ex},
   dp/.code 2 args={%
      \pgfmathsetlengthmacro{\dpheight}{\dpportsep * (max(#1,#2)) * \dpy}
      \pgfkeysalso{draw,%
        minimum width=\dpminwidth,%
        minimum height=\dpheight,%
        font=\bfseries,
        outer sep=0pt,%
        inner sep=5pt,%
        rounded corners=\dpcorners,
        thick,
        prefix after command={\pgfextra{\let\fixname\tikzlastnode}},
        append after command={\pgfextra{\draw
            \ifnum #1=0{} \else foreach \i in {1,...,#1} { 
            ($(\fixname.north west)!{\i/(#1+1)}!(\fixname.south west)$) +(0,0) node[solid,left,circle,color=dpgreen,draw,fill=dpgreen,scale=0.3] {} coordinate (\fixname_fun\i) -- +(0,0) coordinate (\fixname_fun\i')}\fi %
            \ifnum #2=0{} \else foreach \i in {1,...,#2} {
            ($(\fixname.north east)!{\i/(#2+1)}!(\fixname.south east)$) +(0,0) coordinate (\fixname_res\i') -- +(0,0) node[solid,right,circle,color=dpred,draw,fill=dpred,scale=0.3] {} coordinate (\fixname_res\i)}\fi;
         }}}
         },
      dp name/.style={append after command={\pgfextra{\node[label=center,inner sep=2pt,fill=white] at (\fixname) {\textbf{#1}};}}}
   }
\renewcommand{\baselinestretch}{0.925}
\theoremstyle{definition}

\newtheorem*{assumption*}{Assumption}
\newtheorem{theorem}{Theorem}

\newtheorem{lemma}[theorem]{Lemma}

\newtheorem{definition}{Definition}
\theoremstyle{remark}
\newtheorem*{remark}{Remark}

\Crefname{figure}{Fig.}{Figures}
\crefname{equation}{Eq.}{Eqs.}
\makeatletter
    \if@cref@capitalise
        \crefname{subsection}{Section}{Sections}
        \crefname{subsubsection}{Section}{Sections}
        \crefname{assumption}{Assumption}{Assumptions}
        \crefname{problem}{Problem}{Problems}
    \else
        \crefname{subsection}{section}{sections}
        \crefname{subsubsection}{section}{sections}
        \crefname{assumption}{assumption}{assumptions}
        \crefname{problem}{problem}{problems}
    \fi
\makeatother

\usetikzlibrary{positioning,intersections, hobby, patterns, calc, decorations.pathmorphing, decorations.markings, shadows,shapes, cd,arrows.meta, fit,quotes}
\tikzset{
   tick/.style={postaction={
      decorate,
      decoration={markings, mark=at position 0.5 with {\draw[-] (0,.4ex) -- (0,-.4ex);}}}
   }
}
\tikzstyle{block} = [draw, rectangle, minimum height=2em, minimum width=3em,font=\bfseries,rounded corners,thick]
\tikzstyle{block1} = [draw, rectangle, minimum height=1.5em, minimum width=2.5em]
\tikzstyle{blockDyn} = [draw, rectangle, minimum height=2.5em, minimum width=3.5em, align=center, inner sep=10pt, thick, fill=white, copy shadow={draw=black,fill=black,opacity=1,shadow xshift=0.5ex,shadow yshift=-0.5ex}]
\tikzstyle{blockAlg} = [draw, rectangle, minimum height=1.5em, minimum width=2.5em, align=center, inner sep=10pt, thick]
\tikzstyle{sum} = [draw,circle]

\tikzstyle{blockfill} = [block,rounded corners=4,fill=white]

\tikzstyle{nodePre} = [circle, draw,inner sep=1pt,node contents={$\preceq$},thick]
\tikzstyle{nodePreEmpty} = [circle, draw,inner sep=1pt,thick]
\tikzstyle{nodePos} = [circle, draw,inner sep=1pt,node contents={$\posceq$},thick]
\tikzstyle{nodeProd} = [rectangle, draw,inner sep=4pt,node contents={$\times$},rounded corners,thick]
\tikzstyle{nodeSum} = [rectangle, draw,inner sep=4pt,node contents={$\mathbf{+}$},rounded corners,thick]

\definecolor{DPgreen}{rgb}{0.0, 0.5, 0.0}
\definecolor{red}{rgb}{0.75, 0.0, 0.0}

\newif\ifmargincomments %
\margincommentstrue

\newif\ifextendedversion %

\ifmargincomments

\else

\fi

\title{
\textbf{Task-driven Modular Co-design of Vehicle Control Systems}
}
\author{Gioele Zardini, Zelio Suter, Andrea Censi, Emilio Frazzoli
\thanks{
The authors are with the Institute for Dynamic Systems and Control, ETH Z\"urich, Switzerland. ({\tt gzardini@ethz.ch})}
\thanks{
This work was supported by the Swiss National Science Foundation under NCCR Automation, grant agreement 51NF40\_180545.}
}

\begin{document}

\begin{textblock*}{\textwidth}(15mm,18mm) %
\bf \textcolor{NavyBlue}{To appear in the Proceedings of the 2022 IEEE 61th  Conference on Decision and Control}
\end{textblock*}

\maketitle
\begin{abstract}
When designing autonomous systems, we need to consider multiple trade-offs at various abstraction levels, and the choices of single (hardware and software) components need to be studied jointly.
In this work we consider the problem of designing the control algorithm as well as the platform on which it is executed.
In particular, we focus on vehicle control systems, and formalize state-of-the-art control schemes as monotone feasibility relations.
We then show how, leveraging a monotone theory of co-design, we can study the embedding of control synthesis problems into the task-driven co-design problem of a robotic platform.
The properties of the proposed approach are illustrated by considering urban driving scenarios.
We show how, given a particular task, we can efficiently compute Pareto optimal design solutions.
\end{abstract}

\section{Introduction}
The design of embodied intelligence involves the choice of material parts, such as sensors, computing units, and actuators, and software components, including perception, planning, and control routines.
While researchers mostly focused on particular problems and trade-offs related to specific disciplines in robotics, we know very little about the optimal co-design of autonomous systems.
Traditionally, the design optimization of selected components is treated in a compartmentalized manner, hindering the collaboration of multiple designers, and, importantly, modular design automation.
In particular, existing techniques do not allow one to consider simultaneously the specificity and formality of technical results for selected disciplines (e.g., decision making and perception), and more practical trade-offs related to energy consumption, computational effort, performance, and cost of developed robotic systems.
To capture both sides of the coin, one needs a comprehensive task-driven co-design automation theory, allowing multiple domains to interact, focusing both on the component and on the system level~\cite{Censi2015, censi2022}.
What control scheme should one choose to solve a specific task?
Which sensors are really needed to perform state estimation in selected situations?
How does the answer to the previous questions change when considering energetic, computational, and financial constraints?

In this work we propose an approach for the co-design of the control system and of the platform on which it is executed, focusing on the example of \glspl{abk:av}, extending our previous efforts on the subject~\cite{zardiniecc21,zardiniIros21}.

\begin{figure}[t]
\begin{center}
\includegraphics[width=\linewidth]{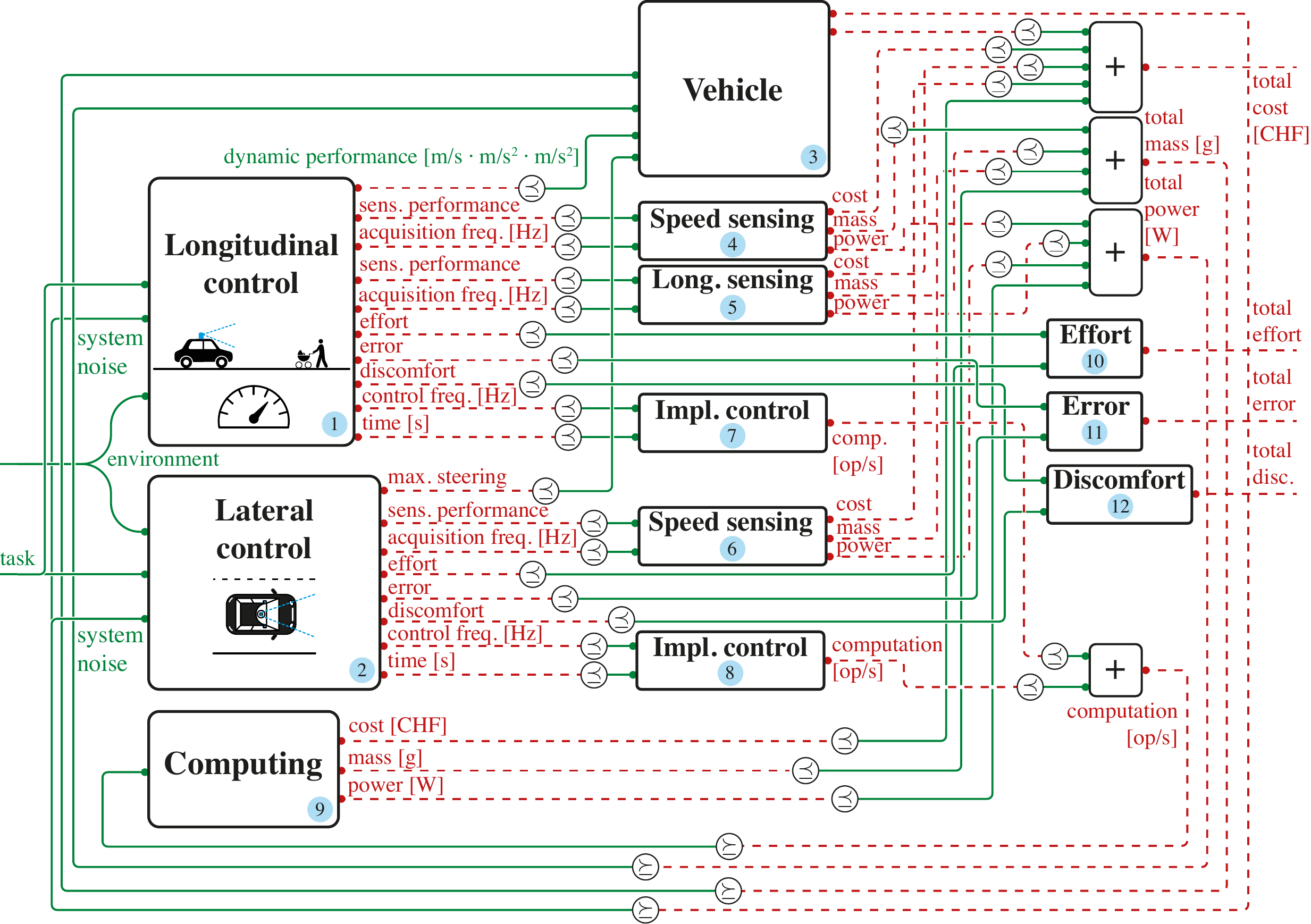}
\end{center}
\caption{In this work we illustrate how to leverage a monotone theory of co-design to embed vehicle control design in a \gls{abk:cdpi} for a full robotic platform.
The proposed results enable designers to formulate and solve the full task-driven co-design problem of a \gls{abk:av}, characterized by choices related to hardware and software components.
By fixing a task (e.g., desired speed and curvature of the trajectory considered), one can find the set of designs fulfilling the task and minimizing selected resources, including quality of execution, performance, effort (discomfort), costs, and safety.}
\label{fig:maincodesigndiag}
\end{figure}

\paragraph*{Related literature}
The need for design automation techniques has been recognized in~\cite{zhu2018codesign, zheng2016cross, Saoud2021, seshia2016design, prorok2021beyond}, and traditionally researchers have been focusing on particular instances of co-design of sensing, actuation, and control~\cite{joshi2008,gupta2006,tzoumas2020,zhang2020,Shell21,smith2021monotonicity,karaman2012,pant2021co,tanaka2015,tatikonda2004,soudbakhsh2013co,collin2019phd, collin2020autonomous,bravo2020,lahijanian2018,okane2008,merlet2005,guo2020actuator, pervan2020algorithmic}.

The trade-offs characterizing sensing and control are predominant in the literature.
Although sensor selection typically cannot be solved in a closed form, it has been shown that particular cases feature enough structure for efficient optimization schemes~\cite{joshi2008,gupta2006}.
In particular, the authors of~\cite{tzoumas2020} study sensing-constrained task-driven LQG control, and authors of~\cite{zhang2020,Shell21, smith2021monotonicity,karaman2012,pant2021co} focus on the design of robots which can solve path planning problem.
While in~\cite{tanaka2015} researchers provide a framework to jointly optimize sensor selection and control, by minimizing the information a sensor needs to acquire, authors of~\cite{tatikonda2004, soudbakhsh2013co} propose techniques for optimal control with communication constraints. 
Methodologies to co-design robot architectures have been investigated in~\cite{collin2019phd, collin2020autonomous,bravo2020}, and robot design trade-offs are presented in~\cite{lahijanian2018,okane2008,merlet2005}.
Finally~\cite{guo2020actuator} study optimal actuator placement in large-scale networks control, and~\cite{pervan2020algorithmic} proposes techniques for algorithmic design for embodied intelligence in synthetic cells.
In conclusion, while existing techniques for robotic platforms co-design highlight important trade-offs, they do not allow one to smoothly formulate and solve co-design problems involving control synthesis.

\paragraph*{Statement of contributions}
In this work we show how to leverage a monotone theory of co-design to frame the design of vehicle control strategies in the task-driven co-design of an entire \gls{abk:av}.
We show how to formulate each relevant control technique as  a \glsfirst{abk:mdpi}, characterized by feasibility relations between tasks, control performance, accuracy, measurements' and system's noises, and intermittent observations schemes.
Furthermore, we instantiate our discoveries in state-of-the-art case studies, in which we perform the task-driven co-design of an \gls{abk:av} in urban driving scenarios, extending previous basic results presented in \cite{zardiniecc21,zardiniIros21}
We illustrate how, given the model of the \gls{abk:av}, a task, and behavior specifications, we can answer several questions regarding the optimal co-design of the full robotic platform.
\section{Monotone Co-Design Theory}
\label{sec:co-design}

\begin{figure}[t]
\begin{subfigure}{\columnwidth}
    \begin{center}
    \begin{tikzpicture}[DP]
            \node[dp={2}{2}] (cnt) {MDPI};
            \draw[runconn, runame={resources}, relres=above,posres=0.9] (cnt_res1){};
            \draw[runconn, runame={}, relres=above,posres=0.9] (cnt_res2){};
            \draw[funconn, funame={functionalities},relfun=above,posfun=1.15] (cnt_fun1){};
            \draw[funconn, funame={},relfun=above,posfun=1.15] (cnt_fun2){};
\end{tikzpicture}
    \subcaption{A \gls{abk:mdpi} is a monotone relation between \glspl{abk:poset} of \F{functionalities} and \R{resources}. \label{fig:mathcodesign}}
    \end{center}
    \end{subfigure}
\begin{center}
\begin{subfigure}[b]{0.3\columnwidth}
\centering
\scalebox{0.8}{\begin{tikzpicture}[DP]
    \node[dp={1}{1}] (f) {$d$};
    \node[dp={1}{1}, right=1cm of f] (g) {$e$};
    \draw[rconn, rcname={}, fcname={}] (f_res1)  to (g_fun1);
    \draw[runconn, runame={}] (g_res1);
    \draw[funconn, funame={}] (f_fun1);
\end{tikzpicture}}
\subcaption{Series.}
\end{subfigure}
\begin{subfigure}[b]{0.3\columnwidth}
\centering
\scalebox{0.8}{\begin{tikzpicture}[DP]
    \node[dp={1}{1}] (f) {$d$};
    \node[dp={1}{1}, below=0.3cm of f] (g) {$e$};
    \draw[runconn, runame={}] (f_res1){};
    \draw[runconn, runame={}] (g_res1){};
    \draw[funconn, funame={}] (f_fun1){};
    \draw[funconn, funame={}] (g_fun1){};
\end{tikzpicture}}
\subcaption{Parallel.}
\end{subfigure}
\begin{subfigure}[b]{0.3\columnwidth}
\centering
\scalebox{0.8}{\begin{tikzpicture}[DP]
    \node[dp={2}{2}] (f) {$d$};
    \draw[runconn, runame={}] (f_res2){};
    \draw[funconn, funame={}] (f_fun2){};
    \draw[rconn,rcname={},fcname={},feedback=1,loos=3] (f_res1) -- ($(f)+(0,4)$) |- (f_fun1);
\end{tikzpicture}}
\subcaption{Loop.}
\end{subfigure}
\label{fig:diagrams}
\caption{\glspl{abk:mdpi} can be composed in different ways.}
\label{fig:dp_def}
\end{center}
\vspace{-5mm}
\end{figure}
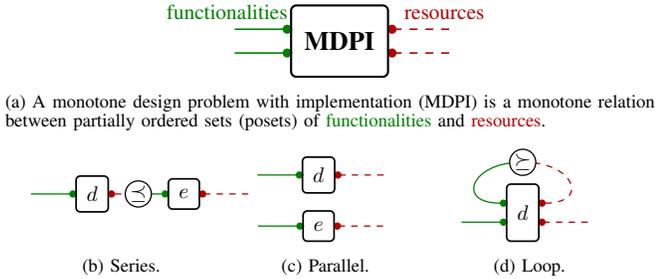
In this section, we summarize the main concepts related to the monotone theory of co-design~\cite{Censi2015, censi2022}, which will be instrumental for the development of the key results in this work.
The reader is assumed to be familiar with basic concepts of order theory~\cite{davey2002} (some of which reported for convenience in \cref{sec:app_order}).
The monotone theory of co-design is based on the atomic notion of a \gls{abk:mdpi}.

\begin{definition}[\gls{abk:mdpi}]
Given \glspl{abk:poset}~$\setOfFunctionalities{},\setOfResources{}$,  (\F{functionalities} and \R{resources}), we define a \emph{\gls{abk:mdpi}} as a tuple~$\tup{\setOfImplementations{d},\prov, \req}$, where~$\setOfImplementations{d}$ is the set of implementations, and~$\prov$,~$\req$ are maps from~$\setOfImplementations{d}$ to~$\setOfFunctionalities{}$ and~$\setOfResources{}$, respectively:
\begin{equation*}
        \setOfFunctionalities{} \xleftarrow{\prov} \setOfImplementations{d} \xrightarrow{\req} \setOfResources{}.
\end{equation*}
We compactly denote the \gls{abk:mdpi} as~$d\colon \setOfFunctionalities{} \tickar \setOfResources{}$.
Furthermore, to each \gls{abk:mdpi} we associate a monotone map~$\bar{d}$, given by:
\begin{equation*}
    \begin{split}
        \bar{d}\colon \setOfFunctionalitiesOp{} \times \setOfResources{} &\to \tup{\power{\setOfImplementations{d}},\subseteq}\\
        \langle \F{f}^*,\R{r}\rangle &\mapsto \{i \in \setOfImplementations{d}\colon (\prov(i) \succeq_{\setOfFunctionalities{}}\F{f}) \wedge (\req(i)\preceq_{\setOfResources{}}\R{r})\},
    \end{split}
\end{equation*}
where~$(\cdot)\op$ reverses the order of a \gls{abk:poset}. 
The expression~$\bar{d}(\F{f}^*,\R{r})$ returns the set of implementations (design choices)~$S\subseteq \setOfImplementations{d}$ for which \F{functionalities}~$\F{f}$ are feasible with \R{resources}~$\R{r}$.
We represent a \gls{abk:mdpi} in diagrammatic form as in \cref{fig:mathcodesign}.
\end{definition}
\begin{remark}[Monotonicity]
Consider a \gls{abk:mdpi} for which we know~$\bar{d}(\F{f}^*,\R{r})=S$.
\begin{itemize}
    \item $\F{f'}\preceq_{\setOfFunctionalities{}} \F{f}\Rightarrow \bar{d}(\F{f'}^*,\R{r})=S'\supseteq S$.
    Intuitively, decreasing the desired functionalities will not increase the required resources;
    \item $\R{r'}\succeq_{\setOfResources{}} \R{r}\Rightarrow \bar{d}(\F{f}^*,\R{r'})=S''\supseteq S$.
    Intuitively, increasing the available resources cannot decrease the provided functionalities.
\end{itemize}
For related examples and detailed explanations we refer to our draft book~\cite{censi2022}.
\end{remark}

\begin{remark}
In practical cases, one can populate the feasibility relations of \glspl{abk:mdpi} with analytic relations (e.g. cost functions, precise relationships), numerical analysis of closed-form relations, and simulations (e.g., POMDPs). For detailed examples refer to~\cite{zardiniecc21, zardiniIros21}.
\end{remark}

\begin{definition}
\label{def:h_map}
Given a \gls{abk:mdpi}~$d$, we define monotone maps
\begin{itemize}
    \item $h_d\colon \setOfFunctionalities{}\to \mathsf{A}\setOfResources{}$, mapping a functionality to the \emph{minimum} antichain of resources providing it;
    \item $h_d'\colon \setOfResources{}{}\to \mathsf{A}\setOfFunctionalities{}$, mapping a resource to the \emph{maximum} antichain of functionalities provided by it.
\end{itemize}
\end{definition}
Solving MDPIs requires finding such maps, in the sense that one of the queries consists in fixing a desired functionality, and finding the set of minimum resources providing it.
Furthermore, if such maps are Scott continuous, and posets involved are complete, one can rely on Kleene's fixed point theorem to find the solution to the queries ``fix a functionality and find minimum resources to achieve it'' and ``fix a resource and find maximum functionalities that can be achieved''.
Individual \glspl{abk:mdpi} can be composed in several ways to form a co-design problem (a multigraph of co-design problems), allowing one to decompose a large problem into smaller subproblems, and to interconnect them.
In \cref{fig:dp_def} we report the main interconnections, and an exhaustive list is presented in~\cite{censi2022}.
Series composition happens when a functionality of a \gls{abk:mdpi} is required by another \gls{abk:mdpi} (e.g., the energy provided by a battery is needed by an electric motor to produce torque).
The symbol~``$\preceq$'' is the posetal relation, which represents a co-design constraint: the resource one problem requires, cannot exceed the functionality another problem provides.
Parallel composition formalizes decoupled processes happening together, and loop composition describes feedback.\footnote{We proved that the formalization of feedback makes the category of \glspl{abk:mdpi} a traced monoidal category~\cite{censi2022}.}
Notably, \glspl{abk:mdpi} are closed under compositions (i.e., a composition of \glspl{abk:mdpi} is an \gls{abk:mdpi}).
We call the composition (multigraph) of \glspl{abk:mdpi} a \glsfirst{abk:cdpi}.
\section{Functional decomposition approach}
In this section, we report the key ingredients to understand our co-design approach, and to appreciate the results presented in \cref{sec:control}~\cite{zardiniIros21,zardiniecc21}.

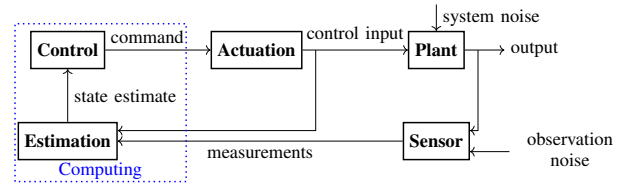
\begin{figure}[tb]
\begin{center}
\scalebox{0.7}{\begin{tikzpicture}
\node[block, rounded corners=0] (cnt) at (0,0) {Control};
\node[block,rounded corners=0, right=2cm of cnt] (act) {Actuation};
\node[block,rounded corners=0, right=2cm of act] (plant) {Plant};
\node[block,rounded corners=0, below=1cm of cnt] (est) {Estimation};
\node[block,rounded corners=0, below=1cm of plant] (sen) {Sensor};
\draw[->] (cnt) -- (act) node[pos=0.4,above]{command};
\draw[->] (act) -- (plant) node[pos=0.5,above]{control input};
\draw[->] (sen) -- (est) node[pos=0.5,below]{measurements};
\draw[->] (est) -- (cnt) node[pos=0.5,right]{state estimate};
\draw[->] ($(sen.east)+(0.75,-0.2)$) -- ($(sen.east)+(0,-0.2)$) node[pos=0,right]{\begin{tabular}{c}observation\\ noise\end{tabular}};
\draw[->] ($(plant.east)+(0,0)$) -- ($(plant.east)+(0.75,0)$) node[pos=1,right]{output};
\draw[->] ($(plant.east)+(0.25,0)$) -- ($(plant.east)+(0.25,-1.5)$)|- ($(sen.east)+(0,0.2)$);
\draw[->] ($(act.east)+(0.25,0)$) -- ($(act.east)+(0.25,-1.5)$)|- ($(est.east)+(0,0.2)$);
\draw[->] ($(plant.north)+(0,0.5)$) -- ($(plant.north)+(0,0)$)node[pos=0.5,right]{system noise};
\draw[dotted,blue,thick] (-1,-2.5) rectangle (2.25,0.5) {};
\node[blue] at (0.625,-2.3) {Computing};

\end{tikzpicture}}
\caption{Classic data-flow diagram for an autonomous system. \label{fig:signalflow}}
\end{center}
\vspace{-0.2cm}
\end{figure}

First, it is important to distinguish between data-flow and logical dependencies.
While \cref{fig:signalflow} represents the classic data-flow diagram for an autonomous system, the co-design diagram in \cref{fig:maincodesigndiag} formalizes logical dependencies through functional decomposition.
In this kind of diagrams, the ``arrows'' are inverted: decision making needs state information, which is estimated through sensing data, provided by a sensor, which will have a cost, a weight, and some power consumption.
The functional decomposition exercise results in a collection of sub-tasks, each of which can be abstracted as a \gls{abk:mdpi} (\cref{fig:dptask}).
Embodied intelligence taks, and in particular control tasks, share a common interface, constituted of task-driven functionalities, such as \F{performance} (e.g., control effort and control error, comfort, safety), and functionalities shared with other parts of the system, such as \F{scenarios/environments} (e.g., robustness to system and measurement noises, density of obstacles). 
Providing functionalities comes at a price, which in robotics can be typically expressed in terms of monetary \R{cost}, \R{power} \R{usage}, \R{computation} \R{effort}, and \R{mass} of the system.

Starting from the task abstraction, one can turn the functional decomposition into a co-design diagram, by interconnecting dependent tasks and identifying common functionalities and resources (\cref{fig:dpfundeco}).

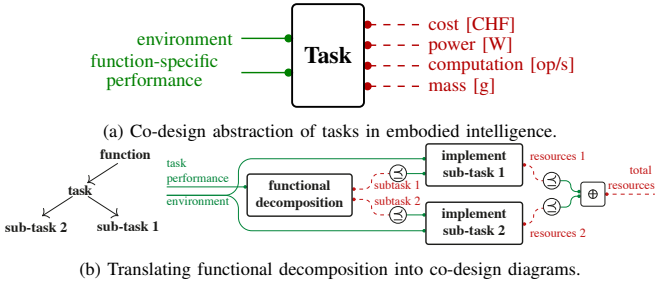
\begin{figure}[t]
\begin{subfigure}[b]{\linewidth}
    \centering
    \scalebox{0.9}{\begin{tikzpicture}[DP, dp port sep = 1.6]
            \node[dp={2}{4}] (cnt) {Task};
            \draw[funconn, funame={environment}, relfun=left] (cnt_fun1){};
            \draw[funconn, funame={\begin{tabular}{c}function-specific\\ performance\end{tabular}}, relfun=left,posres=0] (cnt_fun2){};
            \draw[runconn, runame={cost \unit[]{[CHF]}},relres=right] (cnt_res1){};
            \draw[runconn, runame={power \unit[]{[W]}},relres=right] (cnt_res2){};
            \draw[runconn, runame={computation \unit[]{[op/s]}},relres=right] (cnt_res3){};
            \draw[runconn, runame={mass \unit[]{[g]}},relres=right] (cnt_res4){};
\end{tikzpicture}}
    \subcaption{Co-design abstraction of tasks in embodied intelligence.}
    \label{fig:dptask}
\end{subfigure}
\begin{subfigure}[b]{\columnwidth}
\centering
\includegraphics[width=\linewidth]{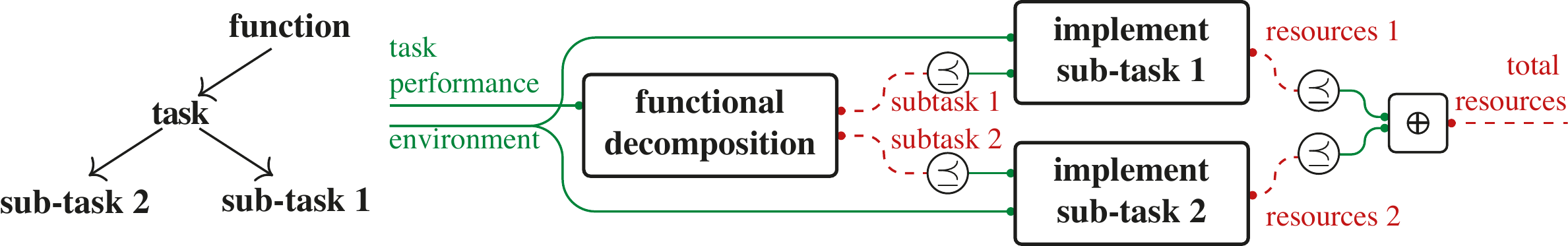}
\subcaption{Translating functional decomposition into co-design diagrams.}
\label{fig:dpfundeco}
\end{subfigure}
\caption{Functional decomposition provides us with sub-tasks, each of which we can model as a \gls{abk:mdpi} with \F{environment} and \F{task performance} as functionalities and \R{cost}, \R{power}, \R{computation}, \R{power}, and \R{mass} as resources. Interconnecting several tasks composing a functional decomposition is a \gls{abk:mdpi} (a task), and resources can be combined via an associative operation, such as~$+$ or~$\max$.}
\label{fig:dpembodiedint}
\vspace{-0.2cm}
\end{figure}

\section{Co-design of vehicle control systems}
\label{sec:control}
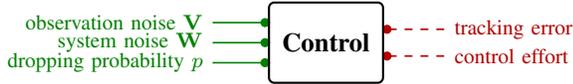
\begin{figure}[tb]
\begin{center}
\begin{tikzpicture}[DP, dp port sep = 1.5]
            \node[dp={3}{2}] (cnt) {Control};
            \draw[runconn, runame={tracking error}, relres=right] (cnt_res1){};
            \draw[runconn, runame={control effort},relres=right] (cnt_res2){};
            \draw[funconn, funame={observation noise $\mat{V}$},relfun=left] (cnt_fun1){};
            \draw[funconn, funame={system noise $\mat{W}$},relfun=left] (cnt_fun2){};
            \draw[funconn, funame={dropping probability $p$},relfun=left] (cnt_fun3){};
\end{tikzpicture}
\end{center}
\caption{General co-design problem for a control problem, highlighting the monotonicity in selected nuisances. For instance, if the probability of dropping observations increases, the tracking error will not decrease.}
\label{fig:codesigntheorems}
\vspace{-0.5cm}
\end{figure}
In~\cite{zardiniecc21} we show how to use the monotone theory of co-design presented in \cref{sec:co-design} to embed (variations of) LQG control design in an \gls{abk:mdpi} for an entire robotic platform.
While finding interesting applications, the presented results were limited to a particular control technique.
In the present work, we extend our theoretical results by showing how one can also embed the design of other important control schemes in a \gls{abk:mdpi}, in the context of designing a \gls{abk:av}.

We present the considered control techniques, as well as their interpretation as \glspl{abk:mdpi} of the form as in \cref{fig:codesigntheorems}.
Theoretical results are sketched in Lemmas, intuitive proof sketches are provided in \cref{sec:proof_sketches}, and exhaustive proofs will be reported in an extended version of this work.
The main idea is to show monotonicities between important quantities for control systems (as shown in \cref{fig:codesigntheorems}), to embed their design in \glspl{abk:mdpi} and solve co-design problems for a robotic platform.

\subsection{Vehicle model}
\paragraph*{Dynamics}
We consider the kinematic single-track model from~\cite{paden2016survey,kong2015kinematic} and extend it by considering model uncertainty.
Consider~$\positionr=[\posxr,\posyr]$ and ~$\positionf=[\posxf,\posyf]$ as the positions of the rear and front wheel with respect to the inertial coordinate frame. 
The heading~$\heading$ describes the vehicle's orientation (the angle between~$\positionf-\positionr$ and the inertial frame).
The steering angle~$\steeringangle$ describes the front wheel orientation with respect to the vehicle's one.
Finally~$l$ is the distance between front and rear axles.
We can write the no-slip condition for the wheels with the following non-holonomic constraints:
\begin{equation*}
\begin{aligned}
-\dot{\posx}_r\sin(\heading)+\dot{\posy}_r\cos(\heading)&=0,\\
-\dot{\posx}_f\sin(\heading + \steeringangle)+\dot{\posy}_f\cos(\heading+\steeringangle)&=0.
\end{aligned}
\end{equation*}
We compute the rear velocity as~$\rearspeed=\dot{\position}_r(\positionf-\positionr)/\Vert \positionf-\positionr\Vert$.
By considering the state space~$\bm{s}=[\posx_\mathrm{r},\posy_\mathrm{r},\heading,\steeringangle,v_\mathrm{r}]$, and control inputs~$\bm{u}=[v_\mathrm{s},a_\mathrm{r}]$, the dynamics read 
\begin{equation}
\label{eq:veh_dyn}
\dot{\bm{s}}=f(\bm{s},\bm{u})+\bm{w}=
\begin{bmatrix}
\rearspeed\cos(\theta)\\
\rearspeed\sin(\theta)\\
\rearspeed\tan(\steeringangle)/l\\
\steerspeed\\
\rearacc
\end{bmatrix}+\bm{w},
\end{equation}
where~$\steerspeed \in [\dot{\steeringangle}_\mathrm{min},\dot{\steeringangle}_\mathrm{max}]$ and~$\rearacc\in [\dot{v}_\mathrm{r,min}, \dot{v}_\mathrm{r,max}]$  are control inputs,~$\steeringangle\in [\steeringangle_\mathrm{min},\steeringangle_\mathrm{max}]\subset [-\pi/2,\pi/2]$,~$\rearspeed\in [v_\mathrm{r,min}, v_\mathrm{r,max}]$, and~$\bm{w}$ is a standard Brownian process with effective noise covariance~$\mat{W}$.
The motion of the front wheel is then:
\begin{equation*}
\begin{aligned}
\dot{\posx}_\mathrm{f}&=\frontspeed\cos(\theta+\steeringangle),\\
\dot{\posy}_\mathrm{f}&=\frontspeed\sin(\theta+\steeringangle),\\
v_\mathrm{f}&=\rearspeed/\cos(\steeringangle).
\end{aligned}
\end{equation*}
Furthermore, the angular velocity is~$\angular=\frontspeed\sin(\delta)/l$.

\paragraph*{Measurement model}
We consider the discrete-time measurement model
\begin{equation}
\label{eq:meas_model}
    \bm{y}_k=\gamma_k(\bm{s}_k+\bm{v}_k),
\end{equation}
where~$\gamma_k\in \{0,1\}$ represents an intermittent observations process (e.g., linear Gaussian Bernoulli, linear Gaussian Markov, and linear Gaussian semi-markov~\cite{censi2010kalman}), and~$\bm{v}_k$ is a standard Brownian process with effective noise covariance~$\mat{V}$ (parametrizing the fidelity of the measurement, i.e., the quality of the sensor).

\begin{remark}
In this work we will use, among others, the Loewner order on the set of Hermitian matrices of order~$n$, denoted~$\mathcal{M}^n$.
Take two Hermitian matrices~$\mat{A},\mat{B}$.
We say~$\mat{A}\preceq_\mathcal{L} \mat{B}$ if and only if~$\mat{B}-\mat{A}$ is positive semi-definite. 
\end{remark}

\begin{definition}
Consider sequences of Hermitian matrices~$\mat{P}$,~$\mat{Q}$ of length~$N\in \mathbb{N}$. 
We define the poset of sequences of Hermitian matrices of length~$N$ by setting~$\mat{P}\preceq_\mathcal{S} \mat{Q}$ if and only if~$\mat{P}_i\preceq_\mathcal{L} \mat{Q}_i$ for all~$i\in N$.
\end{definition}

\subsection{State estimation}
\label{sec:state_est}
Following the literature, we consider an extended Kalman filter (EKF) with the discrete-time measurement model in \cref{eq:meas_model}. 
We summarize the estimation procedure.\\
\noindent \emph{Initialization:}
    \begin{equation*}
        \hat{\bm{s}}(t_0)=\evalue{\bm{s}(t_0)},\ \stateunc{0}{0}{} = \mathbb{E}[(\bm{s}(t_0) - \hat{\bm{s}}(t_0))(\bm{s}(t_0) - \hat{\bm{s}}(t_0))^\intercal].
    \end{equation*} 
\emph{Prediction update:} Solve
\begin{equation*}
\begin{aligned}
\dot{\hat{\bm{s}}}(t)&=f(\hat{\bm{s}}(t),\bm{u}(t)),\\
\dot{\mat{P}}(t)&=\mat{F}(t)\mat{P}(t)+\mat{P}(t)\mat{F}(t)^\intercal+\mat{W}(t),\ \mat{F}(t)=\frac{\partial f}{\partial \bm{s}}\big|_{\hat{\bm{s}}(t),\bm{u}(t)},
\end{aligned}
\end{equation*}
with~$\hat{\bm{s}}(t_{k-1})=\hat{\bm{s}}_{k-1|k-1}$, and~$\mat{P}(t_{k-1}) = \stateunc{k-1}{k-1}{}$, to obtain~$\hat{\bm{s}}_{k|k-1}=\hat{\bm{s}}(t_k)$ and~$\stateunc{k}{k-1}{}=\mat{P}(t_k)$.

\noindent \emph{Measurement update:} Compute
\begin{equation*}
\begin{aligned}
\mat{K}_k&=\stateunc{k}{k-1}{}(\stateunc{k}{k-1}{}+\mat{V})^{-1},\\
\hat{\bm{s}}_{k|k}&=\hat{\bm{s}}_{k|k-1}+\mat{K}_k(\bm{y}_k-\hat{\bm{s}}_{k|k-1}),\\
\stateunc{k}{k}{}&=(\mat{I}-\mat{K}_k)\stateunc{k}{k-1}{},
\end{aligned}
\end{equation*}
where~$\mat{P}$ represents the covariance estimate.

Given this model, one can prove the monotonicity of the covariance estimate with respect to the noise and measurement covariances, and to the probability of loosing observations.
These results will be important when building the control \glspl{abk:mdpi}.
\begin{lemma}
\label{lem:estmon}
The sequence of covariance estimates~$\mat{P}$ produced by the EKF is monotone in~$\mat{V}$ and~$\mat{W}$. In other words:
\begin{equation*}
\tup{\mat{V},\mat{W}}\preceq_{\mathcal{L}\times \mathcal{L}} \tup{\mat{V}',\mat{W}'}\Rightarrow \mat{P}(\mat{V},\mat{W})\preceq_\mathcal{S} \mat{P}(\mat{V}',\mat{W}').
\end{equation*}
\end{lemma}

\begin{lemma}
\label{lem:loosing_ekf}
The sequence of covariance estimates~$\mat{P}$ is monotone in the probability of dropping observations.
\end{lemma}

\subsection{Control}
\label{sec:control_part}
We instantiate the functional decomposition approach for the self-driving task of a \gls{abk:av}. 
In particular, one can decompose this task into lateral and longitudinal control (\cref{fig:fun_dec}).
Longitudinal control can be then split into speed tracking and emergency braking.

\begin{figure}[t]
\centering
\begin{subfigure}[b]{\linewidth}
    \centering
    \includegraphics[width=0.95\linewidth]{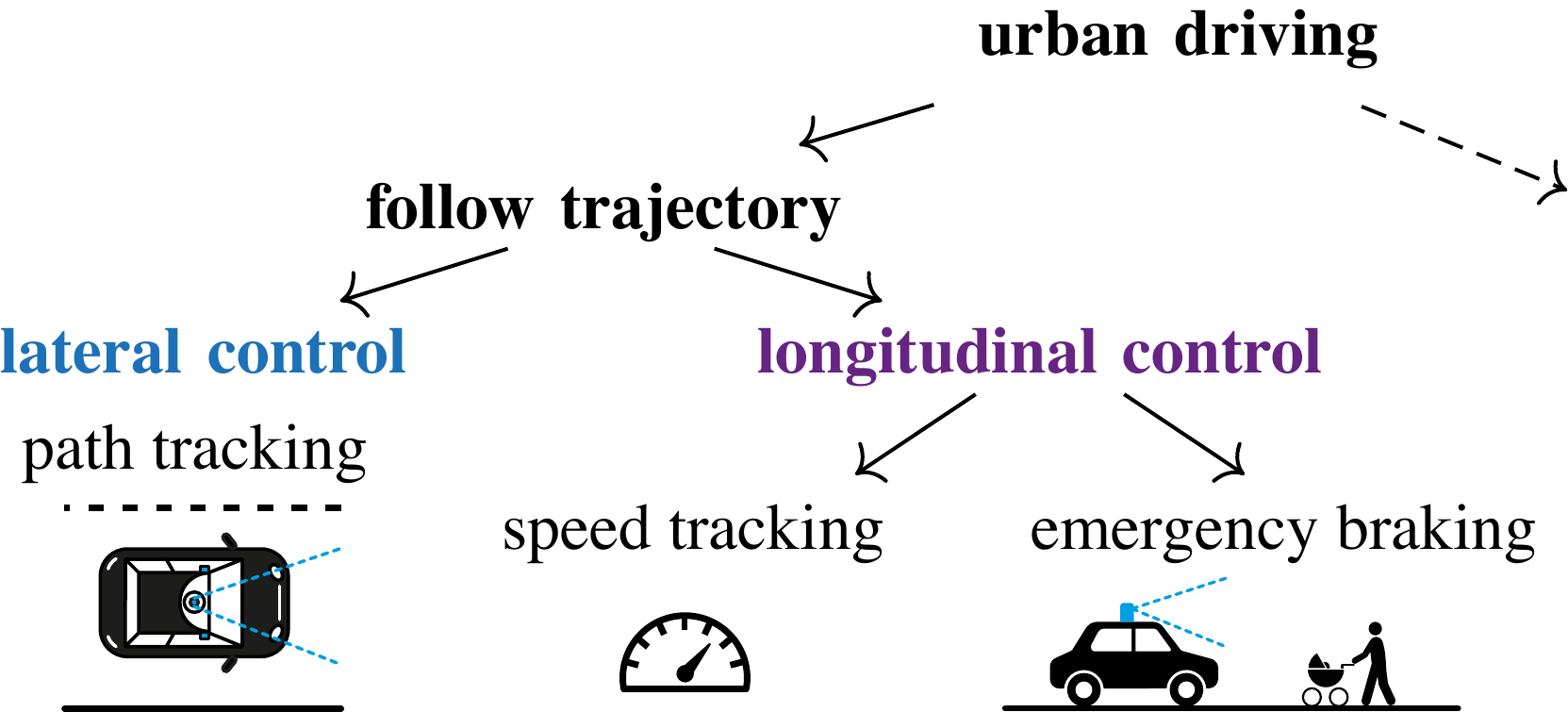}
    \subcaption{Functional decomposition for urban driving.}
    \label{fig:fun_dec}
\end{subfigure}
~
\begin{subfigure}[b]{0.48\linewidth}
    \begin{tikzpicture}
    \node at (0,0) {\includegraphics[width=0.95\linewidth]{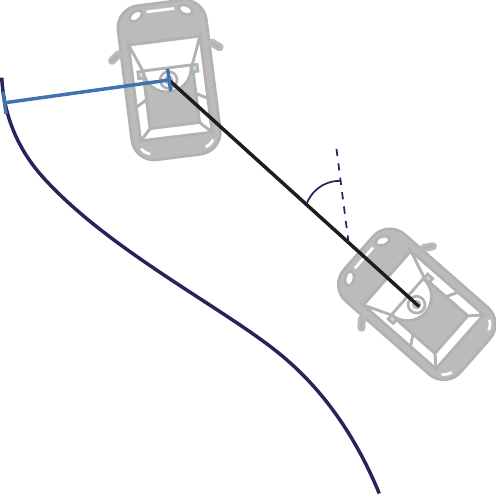}};
    \node at (0.5,0.7) {\small{$\heading_e$}};
    \node at (-1.5,1.45) {\small{$e_p$}};
    \end{tikzpicture}
    \subcaption{Stanley control.}
    \label{fig:stanley}
\end{subfigure}
\begin{subfigure}[b]{0.48\linewidth}
    \begin{tikzpicture}
    \node at (0,0) {\includegraphics[width=\linewidth]{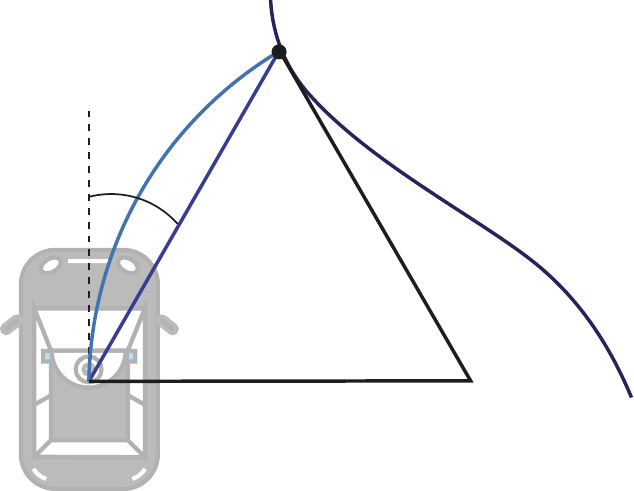}};
    \node at (-1.1,0.1) {\small{$\alpha$}};
    \node at (-0.55,0.4) {\small{$L$}};
    \node at (0.3, 1.25) {\small{$\tup{\posxt,\posyt}$}};
    \end{tikzpicture}
    \subcaption{Pure pursuit control.}
    \label{fig:pure_pursuit}
\end{subfigure}
\caption{(a) Functional decomposition for the self-driving task of a \gls{abk:av}. (b,c): Geometric lateral control with Stanley and pure pursuit techniques.}
\label{fig:pp_stanley}
\end{figure}
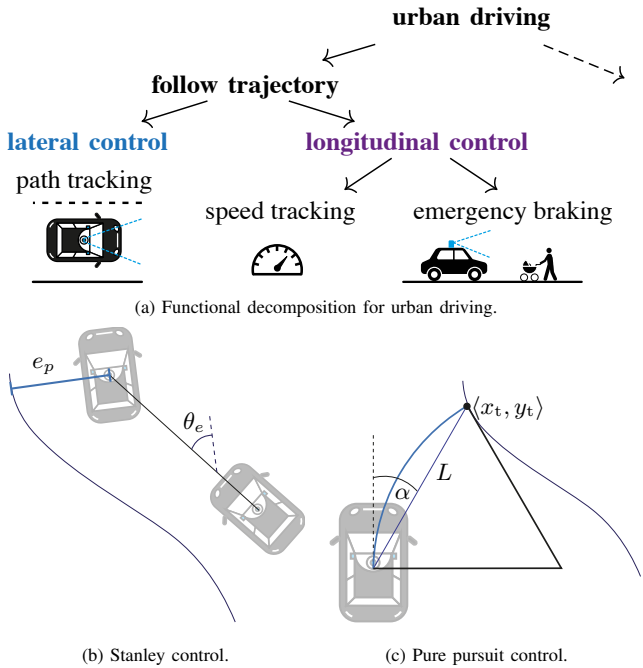
\subsubsection{Lateral control}
The lateral control action can be formulated as choosing the steering velocity to track a given reference path.
The vehicle's speed is typically considered constant at each time step (decoupling the longitudinal and lateral control problems).
Controllers consider the first four states/equations of \cref{eq:veh_dyn} (denoted by~$\bm{z}$) and receive a state estimate at discrete times~$\hat{\bm{z}}_k=\bm{z}_k+\bm{\mu}_k$, where~$\bm{\mu}_k$ is a standard Brownian process with effective noise covariance~$\mat{P}_k$ (as per state estimation procedure and measurement model).
The control error~$\bm{e}(t)=\begin{bmatrix} e_p(t)&\heading_e(t)\end{bmatrix}^\intercal$ is generally expressed as the distance between the vehicle's front axle and the reference point on the path, and the angle between the vehicle's heading and the tangent to the path at the reference point.
In the following, we list the standard lateral control techniques we considered, and their properties~\cite{paden2016survey, rokonuzzaman2021review}.
\paragraph*{Stanley control}
This is a geometric type of vehicle control based on the single-track bicycle model (i.e., the orientation and position of the front wheel with respect to the reference path are considered for generating control actions).
Given the error, one can write the desired steering angle at any time as~$\steeringangle(t)=\heading_e(t)+\arctan(g e_p(t)/v_\mathrm{f})$, where~$g$ is the Stanley gain (\cref{fig:stanley}).
We denote by~$e_{p,\mathrm{tot}}$ the total control positional error and by~$\steeringangle_\mathrm{tot}$ the total control effort along a path.

\begin{lemma}
\label{lem:stanley_mon_error}
The total Stanley control lateral tracking error~$e_{p,\mathrm{tot}}$ is monotonic in~$\mat{W}$ and the sequence of estimate covariances~$\mat{P}$.
\end{lemma}

\begin{lemma}
\label{lem:stanley_mon_effort}
The total Stanley control effort~$\steeringangle_\mathrm{tot}$ is monotonic in~$\mat{W}$ and the sequence of estimate covariances~$\mat{P}$.
\end{lemma}

\paragraph*{Pure Pursuit}
Given a reference path, the control law fits a semi-circle through the vehicle's current configuration to a point on the reference path, which has a distance (called ``lookahead'')~$L$ from the car (\cref{fig:pure_pursuit}).
We consider the algorithm presented in~\cite{coulter1992implementation}, and extend it by requiring the vehicle's heading to be tangent to the circle. 
The curvature of the semi-circle is~$\kappa=2\sin(\alpha)/L$.
Given a constant rear velocity~$\rearspeed$, the angular velocity of a vehicle following the semi-circle is~$\dot{\heading}=2\rearspeed\sin(\alpha)/L$.
From~\cref{eq:veh_dyn}, one has~$\delta=\arctan(2l\sin(\alpha)/L)$, where~$\alpha$ is the angle between the vehicle's orientation and the vector from the current configuration~$\tup{\posx,\posy}$ and the target one~$\tup{\posxt,\posyt}$.
Again, the control error~$\bm{e}(t)=[ e_p(t),\heading_e(t)]$ is expressed as the distance between the vehicle's front axle and the reference point on the path, and the angle between the vehicle's heading and the tangent to the path at the reference point.
The control procedure is then: a) Find current location of the vehicle, b) find the path-point closest to the vehicle, c) find the goal point, d) transform it to the vehicle coordinates, e) calculate the curvature and set the steering angle accordingly, f) update vehicle's position~\cite{rokonuzzaman2021review}.

\begin{lemma}
\label{lem:pp_tracking}
The total pure pursuit control lateral tracking error~$e_{p,\mathrm{tot}}$ is monotonic in~$\mat{W}$ and in the sequence of estimate covariances~$\mat{P}$.
\end{lemma}

\begin{lemma}
\label{lem:pp_effort}
The total pure pursuit control effort~$\steeringangle_{\mathrm{tot}}$ is monotonic in~$\mat{W}$ and in the sequence of estimate covariances~$\mat{P}$.
\end{lemma}
\paragraph*{LQR with adaptive state space control}
The error is given by~$\bm{e}(t)=[ e_p,\heading_e]$, where~$e_p$ is the positional error perpendicular to the path tangent, and~$\heading_e$ is the difference between the path tangent and the vehicle orientation.
The method linearizes the error dynamics around~$[0,0]$ at every time instant, and solves an infinite horizon optimization problem for the linearized system.
The error dynamics for small errors can be formulated as
\begin{equation*}
\dot{\bm{e}}(t)=\begin{bmatrix}
0&\rearspeed\\
0&0
\end{bmatrix}\bm{e}(t)+\begin{bmatrix}
0\\ \rearspeed/l
\end{bmatrix}\steeringangle(t),
\end{equation*}
and the quadratic cost function to minimize takes the form
\begin{equation*}
    J(T)=\int_{0}^T\bm{e}(t)^\intercal \mat{Q}\bm{e}(t)+ \steeringangle^2(t)\mat{R}\text{d}t.
\end{equation*}

\begin{lemma}
\label{lem:lqr_error}
The total LQR control lateral tracking error~$e_{p,\mathrm{tot}}$ is monotonic in~$\mat{W}$ and in the sequence of estimate covariances~$\mat{P}$.
\end{lemma}

\begin{lemma}
\label{lem:lqr_effort}
The total LQR control effort~$\steeringangle_\mathrm{tot}$ is monotonic in~$\mat{W}$ and in the sequence of estimate covariances~$\mat{P}$.
\end{lemma}

\paragraph*{Nonlinear Model Predictive Control (NMPC)}
The NMPC method with receding horizon strategy aims at minimizing the positional error~$\bm{e}_k$ (expressed with respect to the point on the path which is closest to the vehicle's center of mass at instant~$k$) and control effort~$\steeringangle$ over~$n_\mathrm{h}\in \mathbb{N}$ steps.
The formulation of the optimization problem is as follows:
\begin{equation*}
\begin{aligned}
u_k^\star &= \argmin_{U_k}\sum_{i=0}^{n_\mathrm{h}+1}e_{k+i}^\intercal \mat{Q}e_{k+i}+u_{k+i}^\intercal \mat{R}u_{k+i},\\
U_k&=\{u_k,\ldots,u_{k+n_\mathrm{h}-1}\},\\
e_k&=\hat{e}(t_k),\\
e_{k+i}&=\int_{k+i-1}^{k+i}\frontspeed \sin(\theta_e(\tau)-u(\tau))\text{d}\tau,
\end{aligned}    
\end{equation*}
where~$\dot{\heading}_e$ follows \cref{eq:error_stanley}, and where one only applies~$u_0^\star$ each time.
This technique is characterized by different path approximation techniques (e.g., linear, quadratic, and cubic), integration techniques, and by different lateral error reference points on the vehicle (e.g., rear or center of gravity).
\begin{lemma}
\label{lem:nmpc_mon_error}
The total NMPC lateral control tracking error is monotonic in~$\mat{W}$ and the sequence of estimate covariances~$\mat{P}$.
\end{lemma}

\begin{lemma}
\label{lem:nmpc_mon_effort}
The total NMPC lateral control effort is monotonic in~$\mat{W}$ and the sequence of estimate covariances~$\mat{P}$.
\end{lemma}

\subsubsection{Speed control}
The control goal is to track a certain target velocity~$v_\mathrm{t}$.
From \cref{eq:veh_dyn}, the velocity dynamics are~$\dot{v}_\mathrm{r}=a_\mathrm{r}+\bm{w}_{\rearspeed}$.
The system receives an estimation of the current velocity through the measurement model at each time instant~$k$:~$\hat{v}_{\mathrm{r},k}=v_{\mathrm{r},k}+
\rho_k$, where~$\rho_k$ is a standard Brownian process with effective noise covariance~$q_k$ (as per state estimation procedure and measurement model).
The control input is typically formulated via a PID control scheme (i.e., just by choosing specific tuning parameters~$k_\mathrm{p},k_\mathrm{i},k_\mathrm{d}$):
\begin{equation*}
    u(t)=k_\mathrm{p}(v_\mathrm{t}-\hat{v}(t))+k_\mathrm{i}\int_{0}^{t}(v_\mathrm{t}-\hat{v}(\tau))\text{d}\tau -k_\mathrm{d}\frac{\partial \hat{v}}{\partial t}.
\end{equation*}

\begin{lemma}
\label{lem:pid_tracking}
The total PID control tracking error is monotonic in~$\mat{W}$ and the sequence of estimate covariances~$q$.
\end{lemma}

\begin{lemma}
\label{lem:pid_effort}
The total PID control effort is monotonic in~$\mat{W}$ and the sequence of estimate covariances~$q$.
\end{lemma}

\subsubsection{Brake control}
The topic of (emergency) braking has been treated in detail in~\cite{zardiniIros21}, where longitudinal sensors (ordered by their performance, expressed via false positives, false negatives, and accuracy curves) were used to detect potential obstacles. 
Clearly, the more uncertain the obstacle detection, the more potentially dangerous will the braking maneuver be.
We delay the treatment of this particular topic to future works, and refer the interested reader to~\cite{zardiniIros21}.
\subsection{Control as \gls{abk:mdpi}}
We can now combine the results presented for estimation (\cref{sec:state_est}) and for control (\cref{sec:control_part}) to formulate a co-design theorem.%
\begin{theorem}
The presented controllers (PID, Stanley, Pure Pursuit, LQR, and NMPC) can be written as \glspl{abk:mdpi} of the form in \cref{fig:codesigntheorems}.
\end{theorem}
\begin{proof}
The proof of the theorem follows easily from the proofs of the previous lemmas.
\end{proof}

\paragraph*{Discussion}
The presented theorem allows one to frame standard vehicle control systems as \glspl{abk:mdpi}.
This result provides an interface to smoothly include control synthesis in the robot co-design problem.
We remark that theory and results can be generalized to the case of uncertain parameters~\cite{censi2017uncertainty}.

\section{Co-design of an autonomous vehicle}
We now show the ability of the proposed framework to embed aspects of task-driven vehicle control synthesis into the co-design of an entire platform.
The guiding example of this work is the one of a \gls{abk:av}, but the same principle can be generalized to other autonomous systems and abstraction levels~\cite{zardiniAnnuRev22, zardini2022co}.
We first present the setting of the case study, then detail the modeling of the \gls{abk:av} as an interconnection of \glspl{abk:mdpi}, and finally describe selected results to showcase the potential of the approach.

\subsection{Setting}
We consider urban scenarios, extending and customizing the ones proposed in the CommonRoads framework~\cite{althoff2017commonroad}.
We implemented the mentioned autonomy pipelines in our own simulator.
While the proposed approach has been tested on several scenarios (e.g., racing, pursue-evasion, exploration), for exposition purposes we focus on two examples.
We first look at the case in which an \gls{abk:av} needs to perform a \ang{90} degrees, and then look at a lane change example (\cref{fig:90_effort_error}, \cref{fig:lane_cost_error}).
The task of the \gls{abk:av} consists in following a trajectory (with customized curvature severity) at a desired speed.
By fixing a particular task we want to find the autonomy pipeline for the \gls{abk:av} to minimize selected resource usages.

\subsection{Co-design model}
We now consider the co-design diagram for an \gls{abk:av} presented in \cref{fig:maincodesigndiag}, and describe the principles behind each block.
While we explicitly showed how to think about control schemes as \glspl{abk:mdpi}, this can be done for the other blocks as well, and all the conditions can be checked empirically or analytically.
The co-design diagram can be thematically subdivided into control, perception, implementation, and evaluation blocks.
\subsubsection{Control}
Hereafter, we describe the longitudinal \circled{1} and lateral \circled{2} design blocks.
Both provide the fulfillment of a \F{task}, in a specific \F{environment} (e.g., characterized by the time of the day, or the density of obstacles on the road~\cite{zardiniIros21}), and with robustness to a particular \F{uncertainty} in the \gls{abk:av} model.
The vehicle is characterized by a \R{dynamic performance} (e.g., parametrized by the reachable speed, acceleration, and steering angle) \circled{3}.
For the control laws to be implemented, \R{observations} (with particular \R{uncertainties}) from sensors are required (for the longitudinal case, to estimate speed and presence of obstacles, and for the lateral case, to estimate position and heading), which are received at particular \R{frequencies}.
Control techniques will need to be implemented at certain \R{frequencies} and will cause specific \R{control efforts} and \R{errors} (in the longitudinal control case, expressed in terms of \R{velocity deviation}, and in the lateral control case, as \R{lateral deviation}), as well as \R{discomfort} (e.g., intensity of the steering, or gravity of accelerations) and \R{dangerous situations} (e.g., increasing the probability to hit obstacles).
For these blocks, models can be obtained analytically and via simulations.

\subsubsection{Perception}
\F{Observations} are provided to controllers via sensors at particular \F{frequencies}, requiring monetary \R{costs}, \R{power consumptions}, and \R{mass} to be transported (\circled{4},\circled{5},\circled{6}).
For these blocks, models are typically obtained from sensor catalogues, photogrammetry, and simulations.
Particular observation schemes can also be artificially perturbed, and observation dropping schemes can be applied.
For details see~\cite{zardiniIros21}.

\subsubsection{Implementation}
Control routines need to be implemented (\circled{7}, \circled{8}) at certain \F{frequencies} requiring \R{computation}, provided by physical computers \circled{9}.
Computers come at a monetary \R{cost}, \R{power consumption}, and \R{mass} to be transported.
The computation required for particular processes can be estimated via benchmarking (indeed, different implementations of the same concept will require different computational performance).
Models for computers can be derived from catalogues.
\subsubsection{Evaluation}
Different choices for the above thematic blocks will give rise to different outcomes, each characterized by various performance metrics, such as total monetary \R{cost} of the platform, control \R{effort}, control \R{error}, \R{danger}, and \R{discomfort} (\circled{10}, \circled{11},\circled{12}).
Furthermore, the \R{mass} and \R{power} required by the components will be provided by the phyiscal vehicle.
Here, models can be customized to fit specific applications.

\subsection{Co-design results}

\setlength{\tabcolsep}{5pt}
\begin{table*}[t]
\begin{footnotesize}
\begin{center}
\begin{tabular}{lll}
\textbf{Variable} & \textbf{Options}\\
\toprule
PP&$L\in \{0.01,0.05,0.5,1.0,2.0\}$\\
Stanley&$g\in \{0.05,0.1,0.5,1.0,1.5,2.0\}$\\
LQR&$\mat{R}\in \{0.001,0.05,0.5,1.0,10.0\}$, $\mat{Q}\in \{0.1\cdot \mathbb{1},1.0\cdot \mathbb{1},10.0\cdot \mathbb{1},0.2\cdot (\mathbb{1}-0.9e_3),1.0\cdot (\mathbb{1}-0.9e_3),5.0\cdot (\mathbb{1}-0.9e_3)\}$ \\
NMPC&$n_\mathrm{h}\in \{10.0,15.0,20.0,25.0\}$, $\mat{R}\in \{0.05,0.5,1.0,5.0\}$, $\mat{Q}\in \{0.01\cdot \mathbb{1},0.1\cdot \mathbb{1},1.0\cdot \mathbb{1},10.0\cdot \mathbb{1}\}$\\
PID&$k_\mathrm{p}\in \{0.1,0.5,1.0,2.0\}$,~$k_\mathrm{I}\in \{0.01,0.1,0.5,1.0\}$,~$k_\mathrm{d}\in \{0.01,0.05,0.1,1,0\}$\\
Computers&RPi 4B, Jetson Nano/TX1,2/AGX Xavier, Xavier NX\\
Sensors &Basler Ace251gm/222gm/13gm/7gm/5gm/15um, Flir Pointgrey, KistlerSMotion, OS032/128, OS232/128, HDL 32/64\\
\bottomrule
\end{tabular}
\end{center}
\caption{Variables and options for the \gls{abk:av} co-design problem.}
\label{tab:designvars}
\end{footnotesize}
\end{table*}

We now solve the co-design problem presented in the previous section focusing on a selection of queries.\footnote{The solution techniques for this kind of optimization problems and their complexity are described in~\cite[Proposition 5]{Censi2015} and in our talk at \url{https://bit.ly/3ellO6f}.
Once the co-design diagram is created, one can directly solve it via a solver based on formal language.
We are writing a book on the subject, and teaching classes; see \url{https://applied-compositional-thinking.engineering}.}
By fixing a task (i.e., a desired \F{scenario}, \F{speed}, and average \F{curvature}), we can characterize optimal design solutions in terms of monetary \R{cost}, control \R{effort}, control \R{error}, \R{danger}, and \R{discomfort}.
The design space is characterized by the controllers presented in \cref{sec:control} and their parameters, various sensors for the different perception blocks, and computer models, all listed for convenience in \cref{tab:designvars}.
For simplicity of exposition, we do not consider obstacle detection modeling (already treated in depth in~\cite{zardiniIros21}), and focus on path tracking and speed control.
Note that this represents just a sample of the designs we can look at. (For instance, we neglect control frequencies).

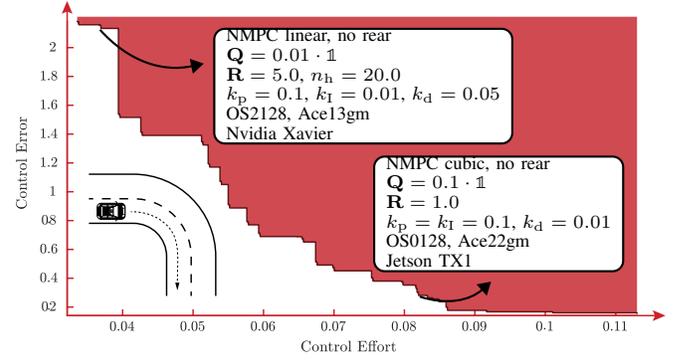
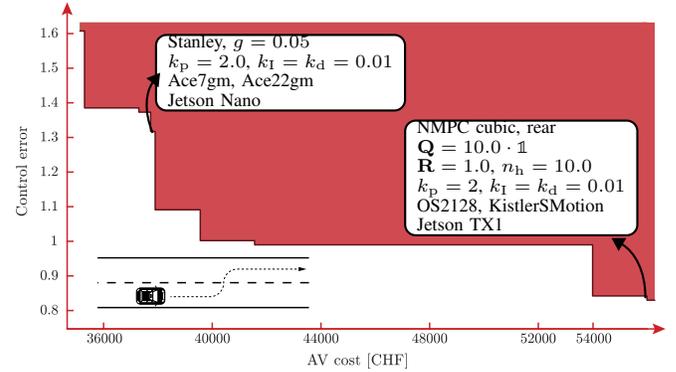
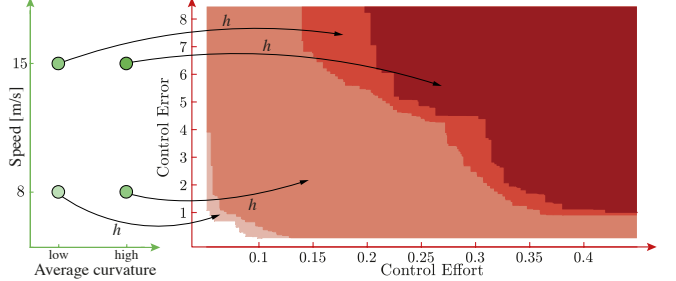
\begin{figure}[t]
    \centering
\begin{subfigure}[b]{\linewidth}
\begin{tikzpicture}
    \node at (0,0){
    \includegraphics[width=\linewidth]{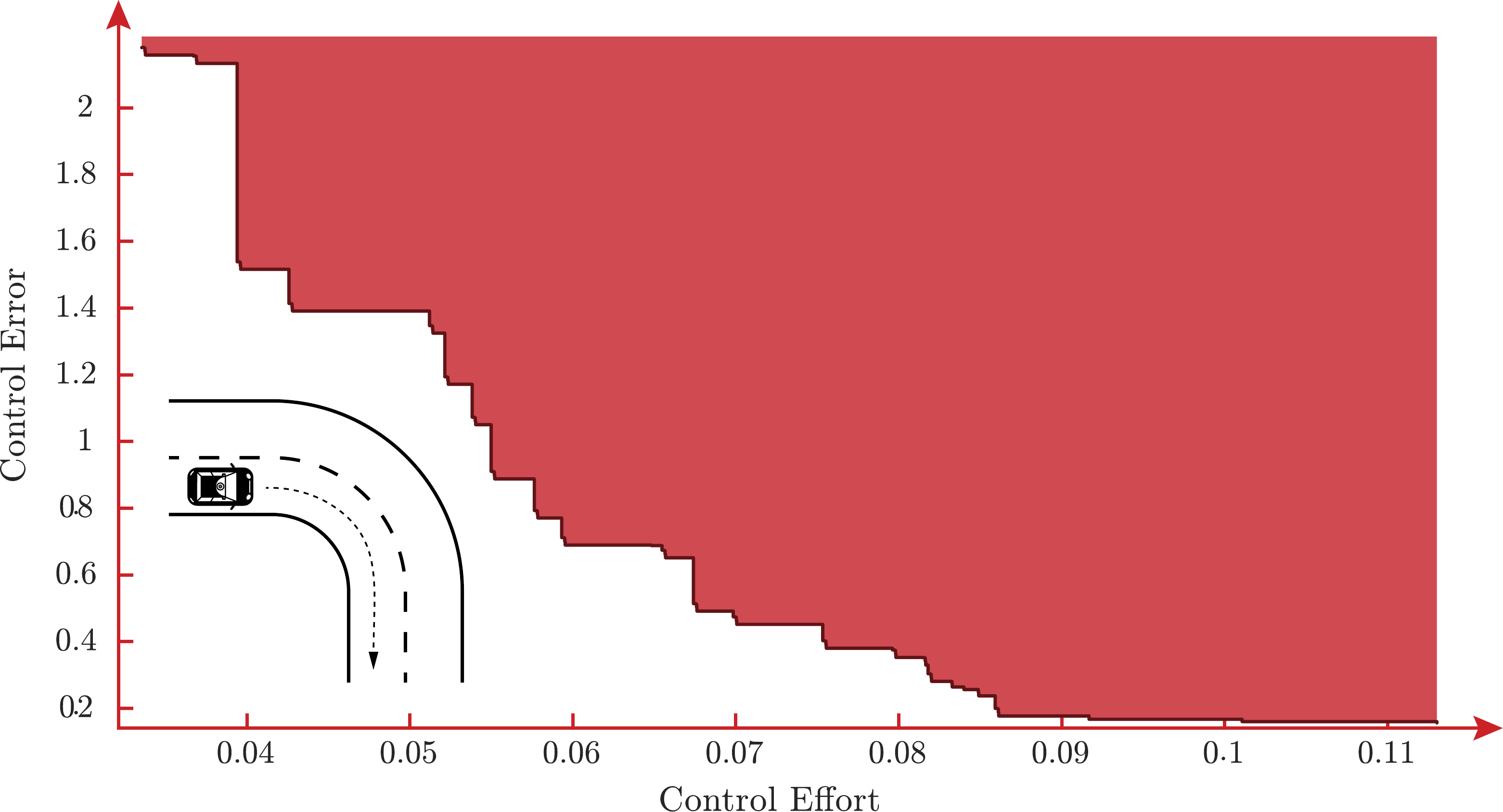}};
    \node[blockfill, inner sep=-0.5pt, font=\scriptsize] at (0.3,1.2) {\begin{tabular}{l}
    NMPC linear, no rear\\
    $\mat{Q}=0.01\cdot \mathbb{1}$\\
    $\mat{R}=5.0$,~$n_\mathrm{h}=20.0$\\
    $k_\mathrm{p}=0.1$,~$k_\mathrm{I}=0.01$,~$k_\mathrm{d}=0.05$\\
    OS2128, Ace13gm\\
    Nvidia Xavier
    \end{tabular}};
    \node[blockfill, inner sep=-0.5pt, font=\scriptsize] at (2.1,-0.5) {\begin{tabular}{l}
    NMPC cubic, no rear\\
    $\mat{Q}=0.1\cdot \mathbb{1}$\\
    $\mat{R}=1.0$\\
    $k_\mathrm{p}=k_\mathrm{I}=0.1$,~$k_\mathrm{d}=0.01$\\
    OS0128, Ace22gm\\
    Jetson TX1
    \end{tabular}};
    \draw[-Triangle,thick] (-3.2,1.95) to[bend right=30] (-1.82,1.5);
    \draw[-Triangle,thick] (1.05,-1.6) to[bend right=30] (2,-1.4);
\end{tikzpicture}
    \subcaption[b]{Trade-off (antichain) of total control error and effort for a \ang{90} turn, with low curvature at \unit[8]{m/s}, with corresponding design choices.}
    \label{fig:90_effort_error}
\end{subfigure}
\begin{subfigure}[b]{\linewidth}
\begin{tikzpicture}
    \node at (0,0){
\includegraphics[width=\linewidth]{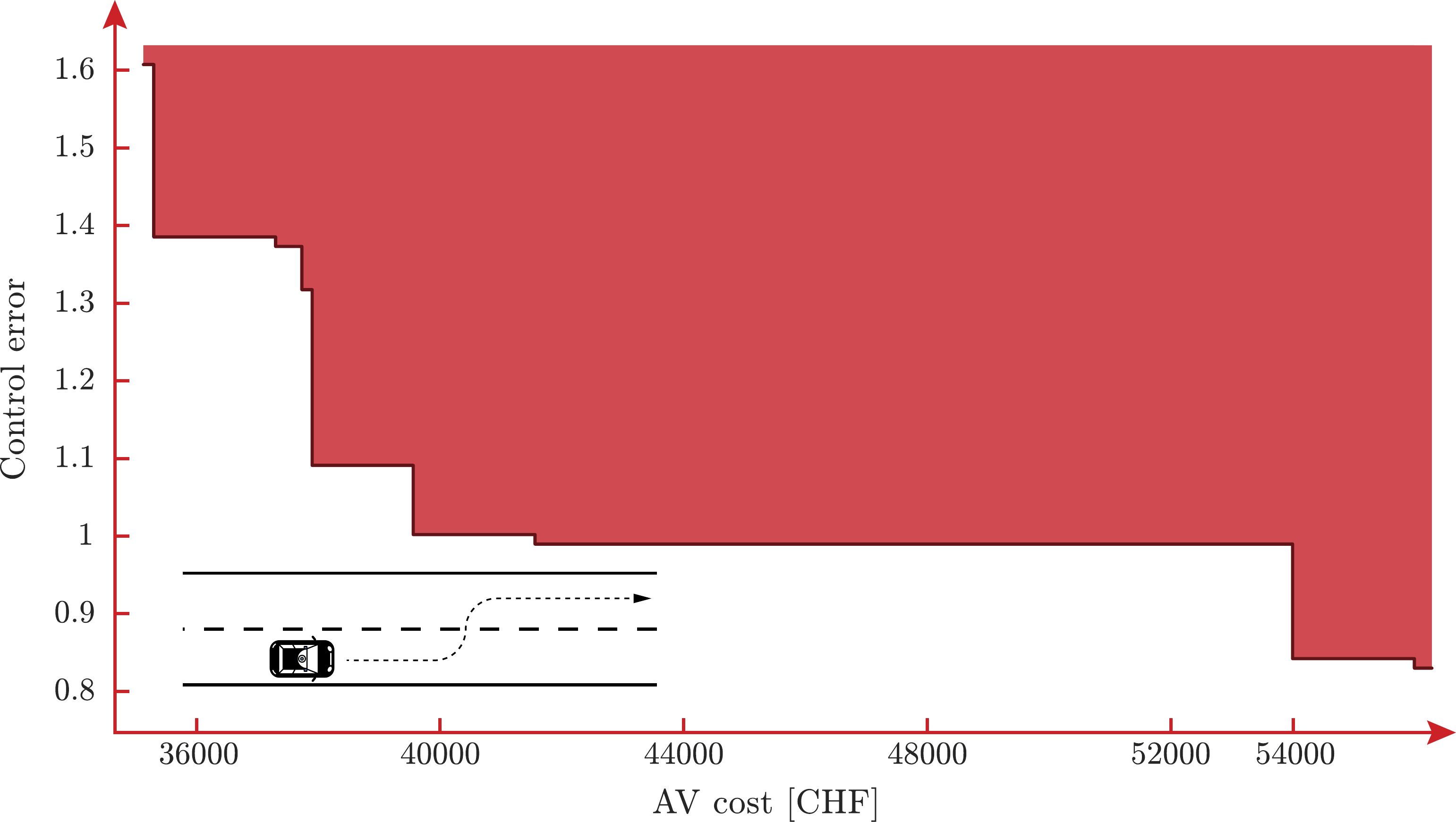}};
\node[blockfill, inner sep=-0.5pt, font=\scriptsize] at (2.4,0.1)
{\begin{tabular}{l}
    NMPC cubic, rear\\
    $\mat{Q}=10.0\cdot \mathbb{1}$\\
    $\mat{R}=1.0$,~$n_\mathrm{h}=10.0$\\
    $k_\mathrm{p}=2$,~$k_\mathrm{I}=k_\mathrm{d}=0.01$\\
    OS2128, KistlerSMotion\\
    Jetson TX1
    \end{tabular}};
\node[blockfill, inner sep=-0.5pt, font=\scriptsize] at (-0.8,1.5) {\begin{tabular}{l}
    Stanley,~$g=0.05$\\
    $k_\mathrm{p}=2.0$,~$k_\mathrm{I}=k_\mathrm{d}=0.01$\\
    Ace7gm, Ace22gm\\
    Jetson Nano
    \end{tabular}};
    \draw[-Triangle,thick] (-2.5,0.7) to[bend left=30] (-2.4,1.5);
    \draw[-Triangle,thick] (4.05,-1.5) to[bend right=30] (3.6,-0.7);
\end{tikzpicture}
    \subcaption{Trade-off (antichain) of cost and control error for lane change with high curvature at \unit[15]{m/s}, with corresponding design choices.}
    \label{fig:lane_cost_error}
\end{subfigure}
\begin{subfigure}[b]{\linewidth}
\includegraphics[width=\linewidth]{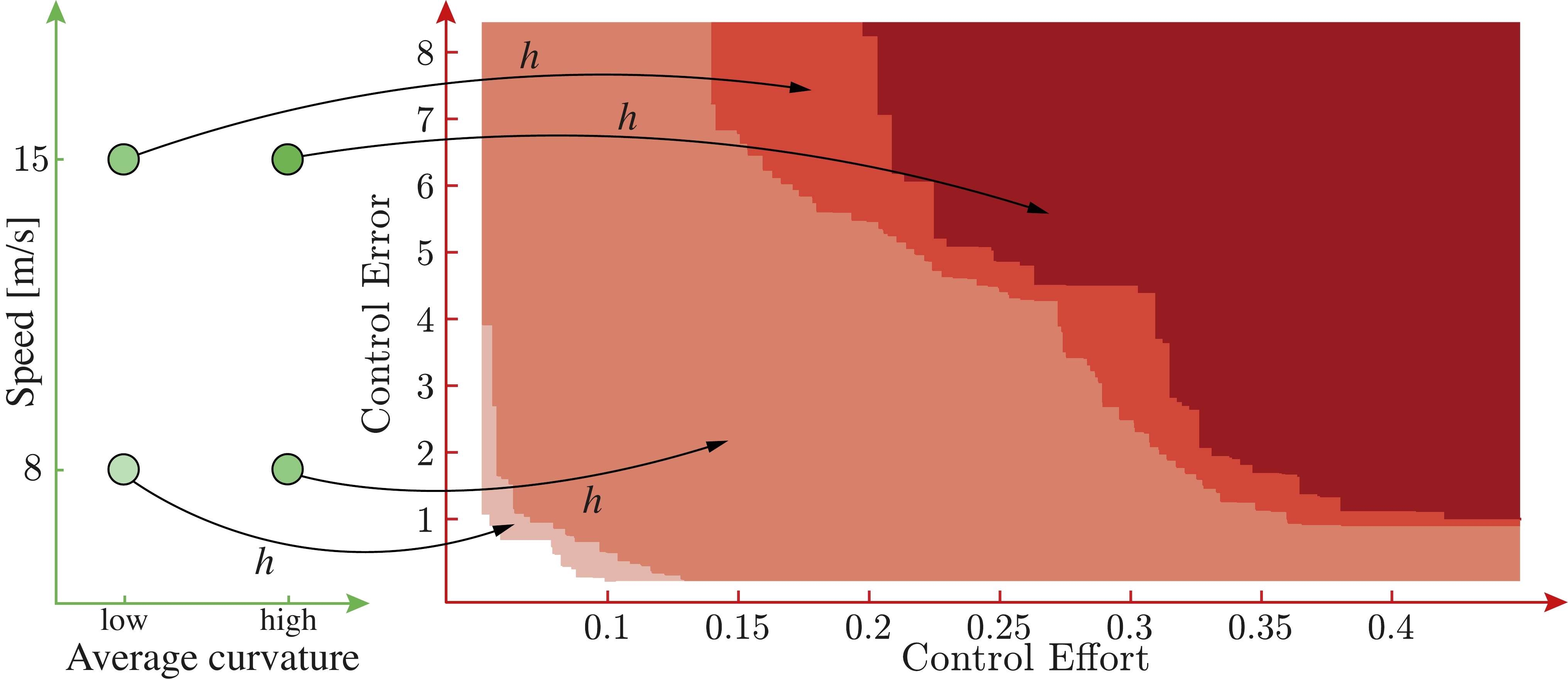}
\subcaption{Monotonicity of the \gls{abk:mdpi}: higher cruise speed or curvature will require higher control error and effort.}
    \label{fig:monotonicity}
\end{subfigure}
\caption{Trade-offs for selected case studies. The red lines represent antichains, and the red areas the corresponding upper sets of resources.}
\end{figure}

\subsubsection{Control effort and control error trade-offs}
We first consider a case in which we want the vehicle to perform a \ang{90} turn with a low curvature, at \unit[8]{m/s}, with a standard battery electric vehicle.
By solving the co-design problem, we obtain a Pareto front of optimal designs, which we can interpret by looking at its 2D-projections.
We first look at the trade-offs between total control effort and total control error (\cref{fig:90_effort_error}).
In red, the Pareto front of optimal solutions, which are not comparable since no instance leads simultaneously to lower control error and control effort.
The upper set of feasible resources is given in solid red.
Furthermore, for each point lying on the Pareto front, we are able to report details about the optimal designs, including the chosen control technique and its parameters, as well as considered sensors and computer.
As one can see in \cref{fig:90_effort_error}, low control effort (discomfort) can be achieved with a specific combination of controllers and parameters, at the cost of an important control error. Similarly, low control error can be achieved by another design, with increased control effort.

\subsubsection{Monetary cost and control error trade-offs}
\label{sec:lane_cost_eff}
We look at the task of lane changing, choosing a high curvature and a speed of \unit[15]{m/s}.
We can now solve the co-design problem with the updated task.
To showcase the richness of the insights we can produce, we now report the trade-offs between monetary costs and control error for the \gls{abk:av} (\cref{fig:lane_cost_error}). 
Clearly, to achieve lower control error one needs to pay more.
Interestingly, investing 39,000 CHF instead of 54,000 CHF only deteriorates the error by 10\%.

\subsubsection{Monotonicity}
We consider the task of lane changing, now showcasing the monotonicity properties of the developed framework.
We look at varying tasks, starting from a low speed of~\unit[8]{m/s} and low curvature and increasing speed to \unit[15]{m/s} and high curvature.
\cref{fig:monotonicity} shows multiple co-design queries.
In particular, for each functionality (left plot), we report the Pareto front and the upper set of optimal resources (right plot), by using the map defined in \cref{def:h_map}.
Monotonicity can be seen in the dominance of subsequent Pareto fronts (right plot), illustrated in increasing red tonality via inclusion of the upper sets.
\section{Conclusions}
We presented how a monotone theory of co-design allows one to simultaneously synthesize state-of-the-art vehicle control systems and design the entire robotic platform.
The proposed approach promotes modularity and compositionality, and captures heterogeneous task-driven design abstraction levels, ranging from control synthesis and parameter tuning to hardware selection.
\paragraph*{Outlook}
The proposed results open the stage for important further developments. 
First, we want to connect the present work with our efforts in modeling sensor curves~\cite{zardiniIros21}, to consider scenarios with moving obstacles and include more granular sensor modeling approaches.
Second, we want to enlarge the library of control schemes which one can embed in the co-design framework (e.g., studying sliding-mode controllers, dynamic switching of controllers), better model digitalisation, include the notion of planning in the models, and present extensive case studies.
Third, we want to leverage our recent results in the co-design theory to query for robustness against multiple, heterogeneous tasks (indeed, one can argue that urban driving will require the ability to solve different ones).
Fourth, we want to study systems which might show complications when embedded in our framework (e.g., Hammerstein-Wiener systems).
Finally, we want to showcase compositionality by connecting the present work to our work in mobility~\cite{zardini2022co}.

\appendix
\label{sec:appendix}
\subsection{Background on orders}
\label{sec:app_order}
\begin{definition}[Poset]
A \emph{\gls{abk:poset}} is a tuple $\mathcal{P}=\langle P,\preceq_\mathcal{P}\rangle$, where $P$ is a set and~$\preceq_\mathcal{P}$ is a partial order, defined as a reflexive, transitive, and antisymmetric relation.
\end{definition}

\begin{definition}[Opposite of a poset]
The \emph{opposite} of a poset~$\langle P,\preceq_\mathcal{P} \rangle$ is the poset $\langle P,\preceq_{\mathcal{P}}^\mathrm{op} \rangle$, which has the same elements as $\mathcal{P}$, and the reverse ordering.
\end{definition}
\begin{definition}[Product poset]
Let~$\tup{P,\preceq_{\mathcal{P}}}$ and $\tup{Q,\preceq_{\mathcal{Q}}}$ be posets. Then,~$\tup{P\times Q,\preceq_{\mathcal{P}\times \mathcal{Q}}}$ is a poset with
\begin{equation*}
    \tup{p_1,q_1}\preceq_{\mathcal{P}\times \mathcal{Q}}\tup{p_2,q_2} \Leftrightarrow (p_1\preceq_{\mathcal{P}}p_2) \wedge (q_1\preceq_\mathcal{Q} q_2).
\end{equation*}
This is called the \emph{product poset} of~$\tup{P,\preceq_{\mathcal{P}}}$ and~$\tup{Q,\preceq_{\mathcal{Q}}}$. 
\end{definition}

\begin{definition}[Monotonicity]
A map~$f\colon \mathcal{P}\to \mathcal{Q}$ between two posets~$\langle P, \preceq_\mathcal{P} \rangle$,~$\langle Q, \preceq_\mathcal{Q} \rangle$ is \emph{monotone} iff~$x\preceq_\mathcal{P} y$ implies $f(x) \preceq_\mathcal{Q} f(y)$. Monotonicity is compositional.
\end{definition}

\subsection{Proof sketches}
\label{sec:proof_sketches}
\begin{proof}[Proof sketch for \cref{lem:estmon}]
One can prove the statement by induction. 
The monotonicity holds in the initialization. 
To prove the induction step, one can first write the prediction update, and leverage properties of the differential Lyapunov equation involved.
Finally, one uses the results of the induction step in the measurement update, to prove the result.
\end{proof}

\begin{proof}[Proof sketch for \cref{lem:loosing_ekf}]
The statement can be proven by following the \emph{substitution principle}~\cite{zardiniecc21}.
If the estimator is given a set of observations, it can simulate having less (i.e., having a higher dropping probability) by artificially ignoring selected samples.
This can also be proven analytically, by comparing measurement updates in the EKF in the two cases.
\end{proof}

\begin{proof}[Proof sketch for \cref{lem:stanley_mon_error}]
First, one can derive the error dynamics
\begin{equation}
\label{eq:error_stanley}
\begin{aligned}
\dot{e}_p(t)&=\frontspeed \sin(\heading_e(t)-\steeringangle(t))+\rho_e(t),\\
\dot{\heading}_e(t)&=-\frontspeed \sin(\steeringangle(t))/L+\bm{w}_\theta,
\end{aligned}
\end{equation}
where~$\rho_e$ and~$\bm{w}_\theta$ are Brownian processes as per given models.
By leveraging properties of the system and measurement noises, one can then show that at any time instant, by larger~$\mat{W}$ or~$\mat{P}$, one cannot obtain a smaller expected total lateral error (to parity of initial condition).
\end{proof}

\begin{proof}[Proof sketch for \cref{lem:stanley_mon_effort}]
The expected lateral positional and orientation errors converging more rapidly to zero imply a commanded steering angle converging more rapidly to zero in expectation. 
This can be formally proven by taking the expectation of the steering angle formulation as presented in the Stanley part.
\end{proof}

\begin{proof}[Proof sketch for \cref{lem:pp_tracking}]
First, one can derive the lateral error dynamics
\begin{equation*}
\dot{e}_p=-\rearspeed e_\mathrm{along}(t)\sin(\steeringangle(t))/L+\rho_e(t),    
\end{equation*}
where~$\rho_e$ is a Brownian process and
\begin{equation*}
e_\mathrm{along}(t)=\begin{bmatrix}
\posxt-\posx(t)&\posyt-\posy(t)
\end{bmatrix}\cdot \begin{bmatrix}
\cos(\heading(t))\\
\sin(\heading(t))
\end{bmatrix}.
\end{equation*}
Then, one can leverage properties of the system and measurement noises to prove the statement.
\end{proof}

\begin{proof}[Proof sketch for \cref{lem:pp_effort}]
This can be shown by following the procedure in \cref{lem:pp_tracking} and looking at the behavior of~$\steeringangle(t)$.
\end{proof}

\begin{proof}[Proof sketch for \cref{lem:lqr_error} and \cref{lem:lqr_effort}]
One can derive the expected error dynamics for the original, nonlinear system, and leveraging properties of the noise perturbations, one can derive both monotonicity results.
\end{proof}

\begin{proof}[Proof sketch for \cref{lem:pid_tracking}]
One can easily write the expression for the expected value of the speed tracking error between any two control steps.
Given the properties of the perturbations in the system model and in the measurement model, one can show that this expression is monotonic in system noise and state estimation uncertainty.
\end{proof}

\begin{proof}[Proof sketch for \cref{lem:pid_effort}]
By explicitly looking at the expression for the expected value of the acceleration resulting from the control, one proceed as for \cref{lem:pid_tracking} and show the monotonicity.
\end{proof}

\begin{proof}[Proof sketch for \cref{lem:nmpc_mon_error} and \cref{lem:nmpc_mon_effort}]
First, one can derive the error dynamics, which are equivalent to \cref{eq:error_stanley}.
One can prove that at each step, the map describing the dependency of the optimal (initial) control input on the initial lateral error is s-shaped (for positive definite~$\mat{Q}$).
Furthermore, in the presence of heading error, the s-shaped curve is translated proportionally along the input axis.
From these facts, one can prove the statements.
\end{proof}

\bibliographystyle{IEEEtran}
\bibliography{paper}

\end{document}